\documentclass[copyright,creativecommons]{eptcs}
\providecommand{\event}{GandALF 2020} % Name of the event you are submitting to

\usepackage{underscore}           % Only needed if you use pdflatex.

\usepackage{amsmath}
\usepackage{amsthm}
\usepackage{amssymb}
\usepackage{stmaryrd}
\usepackage{xspace}
\usepackage{mathpartir}
\usepackage{algorithm, algorithmicx,algpseudocode}
\usepackage{mathtools}
\usepackage{wrapfig}
\usepackage{tikz}
\usetikzlibrary{automata,calc,shapes}

\theoremstyle{plain}
\newtheorem{theorem}{Theorem}
\numberwithin{theorem}{section}
\newtheorem{lemma}[theorem]{Lemma}

\theoremstyle{definition}
\newtheorem{remark}[theorem]{Remark}

\algnewcommand\algorithmicswitch{\textbf{switch}}
\algnewcommand\algorithmiccase{\textbf{case}}
\algnewcommand{\algorithmicglobal}{\textbf{global}}

\algdef{SE}[GLOBAL]{Global}{EndGlobal}[1]{\algorithmicglobal\ #1}{\algorithmicend\ \algorithmicglobal}
\algdef{SE}[SWITCH]{Switch}{EndSwitch}[1]{\algorithmicswitch\ #1\ \algorithmicdo}{\algorithmicend\ \algorithmicswitch}%
\algdef{SE}[CASE]{Case}{EndCase}[1]{\algorithmiccase\ #1}{\algorithmicend\ \algorithmiccase}%
\algtext*{EndProcedure}
\algtext*{EndSwitch}
\algtext*{EndCase}
\algtext*{EndGlobal}%
\algtext*{EndIf}
\algtext*{EndFor}

\newcommand{\Env}{\textsc{env}}
\newcommand{\Mcname}{\textsc{Eval}\xspace}
\newcommand{\Mc}[2]{\ensuremath{\Mcname}(#1,#2)}
\newcommand{\Mcc}[3]{\ensuremath{\Mcname}_{#1}(#2,[#3])}

\newcommand{\States}{\ensuremath{\mathcal{S}}}
\newcommand{\Prop}{\ensuremath{\mathcal{P}}}
\newcommand{\Actions}{\ensuremath{\mathcal{A}}}

\def\newarrow#1{\mathop{{\hbox{\setbox0=\hbox{$\scriptstyle{#1\quad}$}{$%
					\mathrel{\mathop{\setbox1=\hbox to
							\wd0{\rightarrowfill}\ht1=3pt\dp1=-2pt\box1}\limits^{#1}}%
					$}}}}}

\newcommand{\Transition}[3]{\ensuremath{#1 \newarrow{#2} #3}}

\newcommand{\myfalse}{\ensuremath{\mathtt{f\!f}}}

\newcommand{\sem}[3]{\ensuremath{\llbracket #1 \rrbracket^{#2}_{#3}}}

\newcommand{\Var}{\ensuremath{\mathsf{Var}}} %% set of variables
\newcommand{\Func}{\ensuremath{\mathsf{Func}}} %% set of variables

\newcommand{\mydia}[1]{\ensuremath{\langle#1\rangle}} %% one-dim may modality
\newcommand{\myddia}[1]{\ensuremath{\langle\!\langle#1\rangle\!\rangle}} %% one-dim may modality
\newcommand{\mybox}[1]{\ensuremath{[#1]}} %% one-dim must modality

\newcommand{\grtype}{\blacklozenge}
\newcommand{\ord}[1]{\mathsf{ord}(#1)}
\newcommand{\varmon}{+}
\newcommand{\varant}{-}
\newcommand{\varboth}{\pm}

\newcommand{\hfl}{\textup{HFL}\xspace}

\newcommand{\muHO}{\ensuremath{\mu\mathsf{HO}}\xspace}

\renewcommand{\epsilon}{\varepsilon}

\newcommand{\cod}{\texttt{c}}
\newcommand{\blck}{\texttt{b}}
\newsavebox{\textvisiblespacebox}
\savebox{\textvisiblespacebox}{\texttt{x\!\!x}}
\newcommand\vartextvisiblespace[1][\wd\textvisiblespacebox]{%
  \makebox[#1]{\kern.1em\rule{.4pt}{.3ex}%
  \hrulefill%
  \rule{.4pt}{.3ex}\kern.1em}%
}
\newcommand{\spc}{\vartextvisiblespace}

\allowdisplaybreaks

\title{Local Higher-Order Fixpoint Iteration}
\author{Florian Bruse
\institute{School of Electrical Engineering and Computer Science \\ University of Kassel, Germany}
\and
J\"org Kreiker
\institute{Institute for Computer Science \\ University of Applied Sciences Fulda, Germany}
\and
Martin Lange \quad \quad Marco S\"alzer
\institute{School of Electrical Engineering and Computer Science \\ University of Kassel, Germany}
}
\def\titlerunning{Local Higher-Order Fixpoint Iteration}
\def\authorrunning{F.~Bruse, J.~Kreiker, M.~Lange and M.~S\"alzer}
\begin{document}
\maketitle

\begin{abstract}
Local fixpoint iteration describes a technique that restricts fixpoint iteration in function 
spaces to needed arguments only. It has been studied well for first-order functions in abstract interpretation and also in model checking. Here we consider the problem for least and greatest 
fixpoints of arbitrary type order. We define an abstract algebra of simply typed higher-order functions with fixpoints that can express fixpoint evaluation
problems as they occur routinely in various applications, including program verification. We present an algorithm that realises local fixpoint iteration for such higher-order fixpoints, 
prove its correctness and study its optimisation potential in the context of several applications. 
\end{abstract}

% !TEX root =  main.tex

\section{Introduction}
\label{sec:intro}

Fixpoints are ubiquitous in computer science. They serve to explain the meaning of recursion
in programming languages \cite{Winskel93}, database queries \cite{Book:95:Abiteboul:DBFundamentals}, formal 
languages and automata theory \cite{SimoTenn99}; they are being used as logical quantifiers in descriptive
complexity \cite{IB-C992509} or as specialised operators, for instance in temporal logics \cite{TLCS-DGL16}, etc.

Fixpoints often link a denotational and an algorithmic view onto computational problems, most specifically through
Kleene's Fixpoint Iteration Theorem \cite{Kleene:1938}: start with the least, resp.\ greatest value of the underlying 
lattice, and then keep applying the function under consideration until the sequence becomes stable. 
This theorem provides the algorithmic foundations for many applications in which fixpoints
play an important role. For instance in model checking, fixpoint operators are used to describe correctness properties 
\cite{ICALP::EmersonC1980}, and methods based on fixpoint iteration are being used to establish the satisfaction of 
such properties by models of programs \cite{books/daglib/0020348}. 
Fixpoints are used in programming language semantics to explain the meaning of recursive programs. This extends to
static analysis methods. For instance in the original formulation of abstract interpretation
\cite{cousot1977abstract}, collecting semantics extends program semantics to powersets of semantic values ordered by
subset inclusion \cite{cousot1994higher}. Computing program properties then amounts to solving fixpoint equations
over a number of specific (powerset) domains. Fixpoint iteration also provides a standard means for the evaluation of 
recursive database queries \cite{Book:95:Abiteboul:DBFundamentals}.

In many applications, the elements of the lattices in which fixpoints are being sought, are functions themselves. 
In strictness analyses for functional languages \cite{mycroft1980theory,burn1986strictness} for instance, properties
under consideration are sets of functions. Denotational semantics is perhaps the application domain which is most easily 
seen to need lattices of functions, possibly of higher order, in order to explain the meaning of, for example, functional 
programs of higher order. Certain infinite-state model checking problems, in particular so-called higher-order model
checking \cite{conf/lics/Ong06} are tightly linked to the evaluation of fixpoints in functions spaces as well \cite{DBLP:conf/popl/KobayashiLB17}.

We are concerned with the problem of finding fixpoints in such a lattice of functions of some higher order.
Kleene fixpoint iteration in its pure form can still be employed here, but in many situations it is na\"{\i}ve and
inefficient for the following reason. Suppose one is not interested in the entire fixpoint $f$ (which is a function of
some type $M \to N$) but only in the value of this $f$ on some particular argument $x \in M$. Na\"{\i}ve fixpoint 
iteration would start by approximating this function with the least one $f_0 := \underline{\enspace} \mapsto \bot_N$ that 
maps anything to the least element of the lattice $N$, and successively compute better approximations $f_1,f_2,\ldots$ 
until $f_{i+1} = f_i$ for some $i$. Then it would return $f_i(x)$ which equals $f(x)$.

This procedure has then also computed values $f(y)$ for any $y \in M$. It has been observed that a more efficient
approach would be goal-driven and avoid the computation of $f$ on any unnecessary argument. Note that, since $f$ is
defined recursively, the value $f(x)$ may depend on \emph{some} but \emph{not all} values $f(y)$ for $y \in M$. The
term \emph{neededness analysis} was coined to describe the goal-driven evaluation of fixpoints in function lattices, 
avoiding the computation of function values on arguments that do not contribute to the computation of the one value
of interest. 

Neededness analysis has been studied well for lattices of first-order functions as they often arise in abstract 
interpretation \cite{Jorgenson94}. In groundness analyses for logical programs such as
\cite{le1993groundness}, instead of neededness, one rather speaks of a fixpoint computations being \emph{local}
\cite{fecht1996even}, when a solver tries to only compute the values of as few variables as possible.
Neededness analysis has also been studied in the context of model checking complex program properties which cannot be
described in the standard temporal logics of regular expressivity (CTL, $\mu$-calculus) but in extensions using 
predicate transformers \cite{DBLP:conf/lpar/AxelssonL07}. This can be seen as a notion of \emph{on-the-fly} 
model checking for fixpoints of order $1$. Since ``local'' is also a synonym for ``on-the-fly'' in model checking 
\cite{Bhat:1996:ELM}, we stick to the term \emph{local} fixpoint iteration here rather than the more cumbersome 
\emph{neededness analysis}, when referring to a method to avoid the computation of all arguments of fixpoints which
are functions themselves.

In this paper we consider the applicability of local fixpoint iteration in function lattices to arbitrary higher
orders. To this end, we define a simple abstract and typed \emph{higher-order fixpoint algebra} in
Sect.~ \ref{sec:prelim} which can be used to describe evaluation problems involving fixpoints in such lattices. We 
then give a generic local algorithm for evaluating fixpoint terms in higher-order lattices in Sect.~\ref{sec:algo}.
It optimises the na\"{\i}ve fixpoint iteration method by localising the evaluation of recursively defined functions 
at the top order. A formal proof of its correctness is omitted due to space constraints. In Sect.~\ref{sec:applications} we present some computation 
problems which are special instances of the evaluation of higher-order fixpoints in various domains and discuss local
evaluation's optimisation potential by comparing numbers of iteration, resp.\ argument computation steps on some 
hand-crafted examples. In Sect.~\ref{sec:operands} we briefly sketch limitations to the local approach of fixpoint 
iteration for higher-order fixpoints, in the form of obstacles to overcome which do not exist in the first-order
case. We conclude in Sect.~\ref{sec:concl} with an outlook onto further work in this area.

% !TEX root =  main.tex

\section{An Abstract Higher-Order Fixpoint Algebra}
\label{sec:prelim}

\paragraph*{Types and higher-order lattices.} Let $\grtype$ be some base type. Types are derived from the grammar
	\begin{align*}
		\tau &::= \grtype \mid \tau^v \times \dotsm \times \tau^v \to \tau \enspace , \quad v ::= \varmon, \varant, \varboth 
	\end{align*}
where the annotations $v$ are called \emph{variances}, and they specify the dependency of the values of a function of type $\tau_1^{v_1} \times \dotsm \times \tau_n^{v_n} \to \tau$ on their
arguments. In particular, if $v_i = \varmon$, then this dependency is \emph{monotonic}, if $v_i = \varant$ then it is \emph{antitonic}, and it $v_i = \varboth$ then it is unspecified.

The \emph{order} of a type $\tau$ is $\ord{\grtype} := 0$ and $\ord{\tau_1^{v_1} \times \dotsm \times \tau_n^{v_n} \to \tau} := \max \{ \ord{\tau}, \ord{\tau_1}+1,\ldots,$ $\ord{\tau_n}+1 \}$

As usual, a function $f$ on a partially ordered set $(M, \le)$ is \emph{monotonic} if for all $x,y \in M$ with $x \le y$ we have $f(x) \le f(y)$. It is 
\emph{antitonic} if for all such $x,y$ we have $f(y) \le f(x)$. A lattice is a partial order in which suprema and infima, denoted $x \sqcup y$, $x \sqcap y$, resp.\ $\bigsqcup X$ and $\bigsqcap X$
for any $X \subseteq M$ exist, for as long as $X$ is finite.  A lattice is \emph{complete} if these also exist for arbitrary $X$. Complete lattices always contain a least and a
greatest element, usually denoted $\bot$ and $\top$ here. A finite lattice is trivially complete.

Let $\mathcal{M} = (M, \le)$ and $\mathcal{M}_i = (M_i, \le_i)$ for some $i=1,\ldots,n$ be complete lattices. Remember the following constructions on lattices:

\emph{Inverse}: $\mathcal{M}^\varant := (M, \ge)$ where $x \ge y$ iff $y \le x$. For notational convenience we also let $\mathcal{M}^{\varmon} := \mathcal{M}$. It should be clear that these
      two operations not only preserve the property of being a lattice but also completeness.

\emph{Flattening}: $\mathcal{M}^\varboth := (M,=)$, where $=$ denotes equality as usual. Note that $\mathcal{M}^\varboth$ is in general not a lattice anymore, let alone a complete one. 

\emph{Product}: $\prod_{i=1}^n \mathcal{M}_i := (M_1 \times \dotsm \times M_n, \sqsubseteq)$ where $(x_1,\ldots,x_n) \sqsubseteq (y_1,\ldots, y_n)$ iff $x_i \le_i y_i$ for all $i=1,\ldots,n$.
      The product lattice is complete if all its components are complete.

\emph{Higher-order}: $\mathcal{M}_1 \to \mathcal{M}_2 := (\{ f\colon M_1 \to M_2 \mid f$ is monotonic $\}, \sqsubseteq)$ where $f \sqsubseteq g$ if $f(x) \le_2 g(x)$ for all $x \in M_1$.
      The lattice of componentwise ordered monotonic functions from $M_1$ to $M_2$ is complete if $\mathcal{M}_2$ is complete. Completeness of $\mathcal{M}_1$ is not required, not even
      the property of being a lattice since $\le_1$ is not used in the definition of $\sqsubseteq$. 

We can use these constructions to associate, with each type $\tau$, a complete lattice $\mathcal{M}_\tau$, given a complete lattice $\mathcal{M}$ interpreting the ground type $\grtype$:
\begin{displaymath}
\sem{\grtype}{\mathcal{M}}{} := \mathcal{M} \quad, \quad 
\sem{\tau_1^{v_1} \times \dotsm \times \tau_n^{v_n} \to \tau}{\mathcal{M}}{} := \big(\prod_{i=1}^n (\sem{\tau_i}{\mathcal{M}}{})^{v_i}\big) \to \sem{\tau}{\mathcal{M}}{}
\end{displaymath}
Note that each $\sem{\tau}{\mathcal{M}}{}$ is indeed a complete lattice given the remarks above, as the flattening operation that breaks the lattice property is only used on the argument side
of the function operator. Moreover, if $\mathcal{M}$ is finite, then so is $\sem{\tau}{\mathcal{M}}{}$ for all $\tau$.

\paragraph*{Abstract higher-order fixpoint algebra.} Let $\mathcal{M}$ be a complete lattice and $\Func = \{f\colon\tau_f,$ 
$g\colon\tau_g,\ldots\}$ be a set of \emph{computable} and \emph{typed} functions on $\mathcal{M}$, possibly of higher-order. 
Note that if $\tau_f = \grtype$, then $f$ is not really a function but rather a constant. For simplicity we speak of 
functions in this case as well.

Let $\Var := \{x \colon  \tau_x,y\colon\tau_y,\ldots\}$ be a set of typed variables. We write $\tau_x$, 
resp.\ $\tau_f$ for the uniquely determined type of variable $x$, resp.\ function $f$. We will also simply write $x \in \Var$
instead of $(x,\tau_x) \in \Var$ and likewise for the members of $\Func$. 

\emph{Terms} of the abstract higher-order fixpoint algebra over $\Func$, $\muHO(\Func)$ or simply $\muHO$ when $\Func$ is clear 
from the context, are built via 
\begin{displaymath}
\varphi,\varphi_1,\ldots,\varphi_n \enspace ::= \enspace x \mid f \mid \varphi(\varphi_1,\ldots,\varphi_n) \mid \lambda x_1^{v_1},\ldots,x_n^{v_n}.\,\varphi \mid \mu x.\,\varphi \mid \nu x.\,\varphi
\end{displaymath}
where $x_1,\ldots,x_n \in \Var$, $f \in \Func$ and $v_1,\ldots,v_n \in \{\varmon,\varant,\varboth\}$. 

A term $\varphi$ is \emph{closed} if it contains no free variables, where an occurrence of a variable $x$ is free if it is not 
under the scope of some $\lambda \ldots x \ldots$ or $\mu x$ or $\nu x$ in the syntax tree of $\varphi$. In the following, we 
are mainly interested in closed terms; others will usually only occur as subterms of these. Hence, we will often simply speak 
of terms when in fact we mean closed terms at syntactic top-level.

We assume terms to be \emph{well-named}, i.e.\ each variable is bound at most once. Clearly, any term can always be made 
well-named by renaming bound variables.

For better readability, we simply write $\sigma x(y_1^{v_1},\ldots,y_n^{v_n}).\,\varphi$ instead of 
$\sigma x.\,\lambda y_1^{v_1},\ldots,y_n^{v_n}.\,\varphi$, for $\sigma \in \{\mu,\nu\}$. 

\begin{figure}
\begin{mathpar}
\inferrule{\\}{\Gamma \vdash f\colon\tau_f} \and
\inferrule{\Gamma \vdash \varphi\colon \tau_1^{v_1} \times \dotsm \times \tau_n^{v_n} \to \tau \\ \Gamma^{v_1} \vdash \varphi_1\colon \tau_1 \\ \ldots \\ \Gamma^{v_n} \vdash \varphi_n\colon \tau_n}{\Gamma \vdash \varphi(\varphi_1,\ldots,\varphi_n) \colon \tau} \and
\inferrule{v \in \{\varmon,\varboth\}}{\{ x^v, \ldots \} \vdash x\colon\tau_x} \and
\inferrule{\Gamma[x_i^{v_i} \mid i=1,\ldots,n] \vdash \varphi \colon \tau}{\Gamma \vdash \lambda x_1^{v_1},\ldots,x_n^{v_n}.\,\varphi \colon \tau_{x_1}^{v_1}\times \dotsm \times \tau_{x_n}^{v_n} \to \tau} \and
\inferrule{\Gamma[x^\varmon] \vdash \varphi \colon \tau_x}{\Gamma\vdash \sigma x.\,\varphi \colon\tau_x} 
\end{mathpar}
\caption{The typing rules for abstract higher-order fixpoint algebra.}
\label{fig:typerules}
\end{figure}

In order to give terms a well-defined semantics via the Knaster-Tarski Theorem, each $\varphi$ in a term $\mu x.\varphi$ or 
$\nu x.\varphi$ needs to denote a function that is monotone in its argument $x$. Monotonicity is guaranteed for \emph{well-typed terms}, 
to be explained next, and then formally stated as Lemma~\ref{lem:monotonicity} below. Note that the variances are used to track information 
about the monotonicity or antitonicity of functions in particular arguments, and that a monotonic function can be built for instance 
by composing two antitonic ones.
 
A \emph{typing statement} is a triple $\Gamma \vdash \varphi\colon\tau$ where $\varphi$ is a term, $\tau$ is a type, and $\Gamma$ 
is a \emph{typing context} consisting of \emph{typing hypotheses} of the form $x^v$ for $x \in \Var$ and 
$v$ being a variance. For a typing context $\Gamma$, let $\Gamma^\varmon := \Gamma$; % $\Gamma^\varboth := \Gamma$
let $\Gamma^\varant$ result from $\Gamma$ by replacing in it every $x^\varmon$ by $x^\varant$ and vice-versa; 
and let $\Gamma^\varboth = \Gamma^\varmon \cap \Gamma^\varant$, i.e.~the context which only contains typing hypotheses of the 
form $x^\varboth$ from $\Gamma$. The typing context $\Gamma[x^v]$ is obtained by removing $x^{v'}$ from 
$\Gamma$ for any $v'$, and adding $x^v$ instead.

A term $\varphi$ \emph{has type} $\tau$ if the typing statement $\emptyset \vdash \varphi: \tau$ is derivable using the typing rules 
given in Fig.~\ref{fig:typerules}. The rules are standard; they state, for instance, that in function application 
$\varphi(\varphi_1,\ldots,\varphi_n)$, $\varphi$ must have a function type with $n$ arguments which are the types of the respective 
argument terms. Moreover, the arguments themselves have to be typbale in the respective derived typing contexts. For example, if $\varphi$ is antitonic in its first argument, then $\varphi_1$ has to be typable in the typing context $\Gamma^\varant$, where $\Gamma$ is the context used to type the whole application. This reflects the fact that an antitonic function from some lattice is a monotonic function from the inverse of this lattice (cf.~the lattice definitions above and the definition of the semantics below). The rules for fixpoint formulas $\sigma x^\tau.\varphi$ require the term $\varphi$ to be of the same type as $x$,
since being a fixpoint intuitively means $x = \varphi(x)$, and at the same time ensure that $\varphi$ is monotonic in $x$. 
A term is \emph{well-typed}, if it is of some type. 
 
Variance annotations are only used to guarantee well-typedness (and therefore the existence of fixpoints). We will 
always assume that terms are well-typed, and therefore often drop typing annotations for better readability. Note that %, at least for 
for closed terms, 
a unique type for each subterm can easily be recovered. 

\paragraph*{The semantics of terms.} Let $\mathcal{M}$ be a complete lattice, and suppose that all base functions $\Func = \{f\colon\tau_f, \ldots\}$ have an interpretation $f^{\mathcal{M}}$ in the family of higher-order lattices
over $\mathcal{M}$ according to their types. A term $\varphi$ of $\muHO(\Func)$ over $\Func = \{f\colon\tau_f, \ldots\}$ and a set of typed variables $\Var = \{ x\colon\tau_x, \ldots \}$ 
gets interpreted in this family of lattices. In order to explain the value inductively, we need variable interpretations $\eta$ which assign values in lattices over $\mathcal{M}$ to any variable with free 
occurrences in subterms: for each $x\colon\tau_x \in \Var$ we have $\eta(x) \in \sem{\tau_x}{\mathcal{M}}{}$. The value of $\varphi$ over $\mathcal{M}$ and under $\eta$ is denoted 
$\sem{\varphi}{\mathcal{M}}{\eta}$ and is defined inductively as follows.
\begin{align*}
\sem{x}{\mathcal{M}}{\eta} &:= \eta(x) &
\sem{\varphi(\varphi_1,\ldots,\varphi_n)}{\mathcal{M}}{\eta} &:= \sem{\varphi}{\mathcal{M}}{\eta}(\sem{\varphi_1}{\mathcal{M}}{\eta}, \ldots, \sem{\varphi_n}{\mathcal{M}}{\eta}) \\
\sem{f}{\mathcal{M}}{\eta} &:= f^\mathcal{M} &
\hspace*{-5mm} \sem{\lambda x_1^{v_1},\ldots,x_n^{v_n}.\,\varphi}{\mathcal{M}}{\eta} &:= (f_1,\ldots,f_n) \mapsto \sem{\varphi}{\mathcal{M}}{\eta[x_1 \mapsto f_1, \ldots, x_n \mapsto f_n]} \\
\sem{\mu x.\,\varphi}{\mathcal{M}}{\eta} &:= \bigsqcup \{ f \in \sem{\tau_x}{\mathcal{M}}{} \mid \sem{\varphi}{\mathcal{M}}{\eta[x \mapsto f]} \sqsubseteq f \} &
\sem{\nu x.\,\varphi}{\mathcal{M}}{\eta} &:= \bigsqcap \{ f \in \sem{\tau_x}{\mathcal{M}}{} \mid f \sqsubseteq \sem{\varphi}{\mathcal{M}}{\eta[x \mapsto f]} \}
\end{align*} 
The fourth clause, in 
particular its right-hand side, denotes the function that maps a tuple $(f_1,\ldots,f_n)$ of objects from $\sem{\tau_{x_1}}{\mathcal{M}}{} \times \dotsm \times \sem{\tau_{x_n}}{\mathcal{M}}{}$ to
the value $\sem{\varphi}{\mathcal{M}}{\eta[x_1 \mapsto f_1, \ldots, x_n \mapsto f_n]}$ where the subscript index denotes the variable environment that results from $\eta$ by replacing its
bindings for $x_1,\ldots,x_n$ accordingly. For the last two clauses, note that the values on the right-hand side are well-defined according to the Knaster-Tarski Theorem \cite{Tars55} since 
each $\sem{\tau_x}{\mathcal{M}}{}$ is a complete lattice. Note that the semantics of $\muHO$ are easily seen to be invariant under $\beta$-reduction.

\begin{lemma}
\label{lem:monotonicity}
Let $\mathcal{M}$ be a lattice, $\varphi$ be a term of type $\tau'$ under the typing assumptions $\Gamma, x^v$. 
If $ v = \varmon$, then $\sem{\varphi}{\mathcal{M}}{\eta}$ is monotone in $\eta(x)$; if $v=\varant$ then 
$\sem{\varphi}{\mathcal{M}}{\eta}$ is antitone in $\eta(x)$. 
\end{lemma}

\begin{proof} 
By a straightforward induction on the syntax tree of $\varphi$.
\end{proof}

\begin{remark}
\label{rem:kleene}
Over \emph{finite} lattices, each of the type lattices is finite as well. According to Kleene's Fixpoint Theorem, the least and 
greatest fixpoints of a term $\sigma x.\, \varphi$ in $\muHO$ under a variable interpretation $\eta$ can be computed 
by a sequence of approximations as follows: $x_\eta^0 = \top_{\tau_x}$ if $\sigma = \nu$, $x_\eta^0 = \bot_{\tau_x}$ if $\sigma = \mu$, 
and $x_\eta^{i+1} = \sem{\varphi}{\mathcal{M}}{\eta[x\mapsto x_\eta^i]}$. Then, for each finite lattice $\mathcal{M}$ there is 
$n \in \mathbb{N}$ such that $\sem{\sigma x.\, \varphi}{\mathcal{M}}{\eta} = x_\eta^n$. Moreover, these approximations are 
definable in $\muHO$, independently of $\eta$: $\hat{\sigma}$ is defined by $x_\eta^0 = \sigma x.\, x$, and $x_\eta^{i+1}$ is 
defined by the substitution instance $\varphi[x_\eta^i/x]$.
\end{remark}

\paragraph*{Evaluation problems.} We consider the following generic \emph{evaluation} problem: given a (closed) term $\varphi$ of 
$\muHO(\Func)$ with symbols in $\Func$ interpreted in the higher-order lattices over a finite $\mathcal{M}$, compute 
$\sem{\varphi}{\mathcal{M}}{}$. 

This problem is clearly decidable when all basic functions in $\Func$ are computable. A na\"{\i}ve algorithm will simply compute 
the value of each subterm in a bottom-up fashion using Kleene iteration to evaluate fixpoint expressions, and possibly storing 
function values as tables. Note that if $\mathcal{M}$ is finite, so is $\sem{\tau}{\mathcal{M}}{}$ for any $\tau$, but the size 
and height of $\sem{\tau}{\mathcal{M}}{}$ are $k$-fold exponential in the size, resp.\ height of $\mathcal{M}$ when $k = \ord{\tau}$. 

Even for low orders, such a na\"{\i}ve algorithm may perform far too many unnecessary computation steps. Consider the following 
special \emph{local} variant of the evaluation problem: given a finite complete lattice $\mathcal{M}$, a closed term 
$\varphi_0 := \mu x.\,\varphi$ of type $\tau^v \to \grtype$ (which is then necessarily the same as $\tau_x$) for some 
$v,\tau,\varphi$, and a term $\psi$ of type $\tau$, compute $\sem{\varphi_0(\psi)}{\mathcal{M}}{}$.

Note how this problem formulation describes a situation in which na\"{\i}ve fixpoint iteration obviously performs too many 
evaluation steps in general: it computes $\sem{\varphi_0}{\mathcal{M}}{}$ using Kleene iteration which results in a function 
of type $\tau^v \to \grtype$. Depending on the order of $\tau$, this function is huge in terms of its arguments but still finite. 
We would then also compute $\sem{\psi}{\mathcal{M}}{}$. Then we obtain the value $\sem{\varphi_0(\psi)}{\mathcal{M}}{}$ by 
application, for instance through a simple look-up in the table representing $\sem{\varphi_0}{\mathcal{M}}{}$, where 
$\sem{\psi}{\mathcal{M}}{}$ occurs as some argument. Clearly, the value of $\sem{\varphi_0}{\mathcal{M}}{}$ on all other 
arguments is irrelevant, and the reason for their computation is questionable.

\paragraph{A (first-order) example.} Consider the Boolean lattice $\mathbb{B} = \{\top,\bot\}$ and the normal Boolean functions 
$\Func_{\mathsf{bool}} = \{ \wedge,\vee\colon \grtype^\varmon \times \grtype^\varmon \to \grtype, \neg\colon \grtype^\varant \to \grtype, 0,1\colon \grtype \}$ interpreted in the standard way. Let $n > 0$ and
\begin{align*}
\varphi_n \enspace := \enspace \mu F(\underbrace{x_0,\ldots,x_{n-1}}_{\vec{x}}).\ \mathsf{null}(\vec{x}) \vee
&\big(\mathsf{even}(\vec{x}) \wedge F(\mathsf{half}(\vec{x}))\big) \\[-5mm] \vee &
\big(\neg\mathsf{even}(\vec{x}) \wedge F(\mathsf{add}(\mathsf{add}(\vec{x},\mathsf{dbl}(\vec{x})),(1,0,\ldots,0)))\big)\ .
\end{align*}
over $\Var = \{ x_0,\ldots,x_{n-1}\colon \grtype, F\colon (\grtype^\varboth)^n \to \grtype \}$ where, for $\vec{x} = (x_0,\ldots,x_{n-1})$ and $\vec{y} = (y_0,\ldots,y_{n-1})$,
\begin{itemize}
\item $\mathsf{null}(\vec{x})$ returns $\top$ iff $\vec{x}$ encodes the numerical value $0$, for instance $\mathsf{null}(\vec{x}) = \bigwedge_{i=0}^{n-1} \neg x_i$;
\item $\mathsf{even}(\vec{x})$ returns $\top$ iff $\vec{x}$ encodes an even number, for instance $\mathsf{even}(\vec{x}) = \neg x_0$;
\item $\mathsf{half}$ and $\mathsf{dbl}$ represent the operations ``$\div 2$'' and ``$\cdot 2$'' on bit strings;
\item $\mathsf{add}(\vec{x},\vec{y})$ yields a bit string representing the addition of the two values modulo $2^n$.
\end{itemize}
It is not difficult to find Boolean functions realising these operations.

Intuitively, $\varphi_n$ defines a search procedure. Note that any value of $\vec{x}$ encodes a number in the range $[2^n] = \{0,\ldots,2^n-1\}$ which we will simply denote $\vec{x}$ as well. For any
value $\vec{x}$, define a sequence $(\vec{x}_k)_{k \ge 0}$ via $\vec{x}_0 = \vec{x}$, $\vec{x}_{k+1} = \vec{x}_k / 2$ if $\vec{x}$ is even, and $\vec{x}_{k+1} = 3\cdot \vec{x}_k + 1$ if 
$\vec{x}$ is odd.
Hence, suppose $\vec{b} \in \{0,1\}^n$ encodes such a number, then $\varphi_n(\vec{b})$ is true iff this sequence eventually hits the value $0$.

Let $n = 3$. The graph on the right depicts the sequence on all values in $[2^3]$. Assuming 
\begin{wrapfigure}{r}{2.6cm}
\vspace*{-4mm}
\begin{tikzpicture}
  \foreach \x in {0,...,7}
    \node (\x) at (360/8*\x: 1cm) {$\x$};
  
  \path[->] (0) edge [loop above]  ()
            (1) edge   (4)
            (2) edge              (1)
            (3) edge              (2)
            (4) edge [bend left=60]  (2)
            (5) edge   (0)
            (6) edge  (3)
            (7) edge              (6);
\end{tikzpicture}
\vspace*{-4mm}
\end{wrapfigure}
that the Boolean function $\wedge$ only evaluates its second argument when the first one is not $\bot$, this graph suggests how local fixpoint iteration of this first-order function $F$
can be more efficient when the value of the fixpoint $F$ is only needed on one particular argument. The effect of global fixpoint iteration is depicted in Fig.~\ref{fig:example} (left). Here, the iteration
starts with the least function $F^0: \underline{\enspace} \mapsto \bot$, and it terminates when the current approximation equals the last one. 

Local fixpoint iteration on the other hand only adds the arguments successively to those tables. Consider the case of evaluating $\varphi_3(1,0,1)$ which means iterating the numerical series beginning
at $5$. First we only tabulate the approximant $F^0$ on the value under consideration, i.e.\ $5$. In order to compute $F^1(5)$ we need $F^0(0)$, so $0$ gets added as a new argument and receives the
initial value $\bot$ there. Then we compute $F^2(5)$ and $F^1(0)$, and so on. The iteration stabilises when no change is being recorded anymore, thus computing values as they are shown in the table
of Fig.~\ref{fig:example} (middle). 

The effect of computing $\varphi_3(1,1,0)$, i.e.\ beginning the numerical series at $3$ is similar. Here, however, $0$ is never reached. Hence, all values encountered are $\bot$, and the iteration 
stabilises when no further arguments are needed, as shown in Fig.~\ref{fig:example} (right). 

Even though the local iteration computes a value of $F^4$ while the global one only reaches $F^3$, it should be clear that local evaluation performs fewer computation steps in general.

\begin{figure}
\begin{tabular}{r|c|c|c|c|c|c|c|c}
           & $0$    & $1$    & $2$    & $3$    & $4$    & $5$    & $6$    & $7$    \\ \hline\hline
$F^0$:     & $\bot$ & $\bot$ & $\bot$ & $\bot$ & $\bot$ & $\bot$ & $\bot$ & $\bot$ \\ \hline
$F^1$:     & $\top$ & $\bot$ & $\bot$ & $\bot$ & $\bot$ & $\bot$ & $\bot$ & $\bot$ \\ \hline
$F^2$:     & $\top$ & $\bot$ & $\bot$ & $\bot$ & $\bot$ & $\top$ & $\bot$ & $\bot$ \\ \hline
$F^3$:     & $\top$ & $\bot$ & $\bot$ & $\bot$ & $\bot$ & $\top$ & $\bot$ & $\bot$ \\ \multicolumn{1}{c}{ } 
\end{tabular}
\hfill
\begin{tabular}{r|c|c}
           & $5$    & $0$  \\ \hline\hline
$F^0$:     & $\bot$ & $\bot$ \\ \hline
$F^1$:     & $\bot$ & $\top$ \\ \hline
$F^2$:     & $\top$ & $\top$ \\ \hline
$F^3$:     & $\top$ &        \\ \multicolumn{1}{c}{ }
\end{tabular}\hfill
\begin{tabular}{r|c|c|c|c}
           & $3$    & $2$    & $1$    & $4$    \\ \hline\hline
$F^0$:     & $\bot$ & $\bot$ & $\bot$ & $\bot$ \\ \hline
$F^1$:     & $\bot$ & $\bot$ & $\bot$ & $\bot$ \\ \hline
$F^2$:     & $\bot$ & $\bot$ & $\bot$ &        \\ \hline
$F^3$:     & $\bot$ & $\bot$ &        &        \\ \hline   
$F^4$:     & $\bot$ &        &        &   
\end{tabular}
\caption{Global (left) vs.\ local fixpoint iteration for a first-order function $F$.}
\label{fig:example} 
\end{figure}

\begin{remark} 
It is well-known that fixpoint iteration does not need to record the entire history of its computation but, for each variable, merely the value of the last iteration. In the left table (with only one fixpoint variable) in Fig.~\ref{fig:example}, 
this corresponds to two successive rows: the upper for the last approximate value of a function, and the lower for the current value. In the two tables on the right using local iteration, this
corresponds to diagonals, but the picture is more complicated in general, for instance for fixpoint terms in which the fixpoint variable has multiple occurrences. 

The tables shown in Fig.~\ref{fig:example} not only give an idea of how local fixpoint iteration works, their width and height are also good measures for the space, resp.\ time
needed to compute such a higher-order fixpoint.
\end{remark}

% !TEX root =  main.tex

\section{Local Fixpoint Evaluation for Full \texorpdfstring{$\muHO$}{muHO}}
\label{sec:algo}
 
\paragraph*{The algorithm.}
Procedure $\Mcname$ in Alg.~\ref{alg:ahofa-mc} solves the evaluation problem for $\muHO$ terms of arbitrary higher order and
finite lattices using local fixpoint iteration. It takes four parameters:
(1) a term $\varphi \in \muHO(\Func)$ over some $\Func$. It is not necessarily of type $\grtype$, but the algorithm is assumed to be started with a full list of arguments (see below) in order to realise local fixpoint iteration.
(2) A finite, and, hence, complete lattice $\mathcal{M}$ with an interpretation of any $f\colon\tau_f \in \Func$ as an object in $\sem{\tau_f}{\mathcal{M}}{}$. 
(3) A list $T_1,\dotsc,T_k$ of arguments to $\varphi$. The following invariant is maintained: if the type\footnote{Variances are not important in this section. In order to reduce clutter, we do not display them.} of $\varphi$ is $\tau_1\to\dotsb\to\tau_k\to \grtype$, then $T_i \in \sem{\tau_i}{\mathcal{M}}{}$ for all $1 \leq i \leq k$. 
(4) A global variable $\Env$ that is used to interpret free variables. Values of $\lambda$-bound variables are stored as full functions\footnote{This might appear wasteful at first, but $\lambda$-bound variables are never of the highest type (by order) that occurs in the term to be evaluated except in pathological cases, which can be eliminated by $\beta$-reduction before calling $\Mcname$.}, values of fixpoint variables may be stored as partial 
 approximations as described at the end of the previous section. 

In order to bridge the gap between a variable assignment $\eta$, which assigns a value to each variable which is defined at every 
argument, and the global variable $\Env$ which only stores partial approximations for fixpoint-bound variables, consider the following definition. It turns a state of $\Env$ into a well-defined variable assignment:
\[
\eta_{\Env}(x)(T_1,\dotsc,T_k) = \begin{cases} 
    \Env(x)(T_1,\dotsc,T_k) &,\text{ if } x \text{ is } \lambda\text{-bound} \\ 
    \Env(x)(T_1,\dotsc,T_k) &,\text{ if } x \text{ is fixpoint-bound and } \Env(x)(T_1,\dotsc,T_k) \text{  is defined} \\
    \hat{\sigma}_x &, \text{ otherwise}
\end{cases}
\]
Here, $\hat{\sigma}_x$ is $\bot_{\grtype}$, resp.~$\top_{\grtype}$ for variables that are bound to a least, resp.~greatest fixpoint. %, and $\top_{\grtype}$ for variables that are bound to a greatest fixpoint.

Note that, due to the invariants, a state of Alg.~\ref{alg:ahofa-mc}, i.e.~a call of $\Mcname(\varphi,T_1,\dotsc,T_k)$ with a value of $\Env$ and over some higher-order lattice $\mathcal{M}$, can be thought of as computing the object $\sem{\varphi}{\mathcal{M}}{\eta_{\Env}}(T_1,\dotsc,T_k)$, which is always a member of $\sem{\grtype}{\mathcal{M}}{}$. The algorithm computes this value recursively by descending through the syntax tree of $\varphi$. Fixpoints are resolved by Kleene iteration until the semantics computed stabilises, but the value is only computed at the arguments indicated plus all those arguments that are discovered as necessary to obtain the value for the original argument.

\begin{algorithm}[t]
   \begin{tabbing}
     \textbf{procedure} $\Mc{\varphi}{T_1,\ldots,T_k}$: \hspace*{3.8cm} $\vartriangleright$ global (partial) $\Env: \Var \to\bigcup_{x \in \Var} \sem{\tau_x}{\mathcal{M}}{}$ \\
     \quad \= \textbf{switch} $\varphi$: \\
     \> \quad \= \textbf{case} $f$: \hspace{2cm} \= \textbf{return} $f(T_1,\ldots,T_k)$ \\
     \> \> \textbf{case} $x$: \> \textbf{if} $\Env(x)(T_1,\ldots,T_k) = \texttt{undef}$ \textbf{then} $\Env(x):=\Env(x)[(T_1,\ldots,T_k) \mapsto \hat{\sigma}_x]$ \\
     \> \> \> \textbf{return} $\Env(x)(T_1,\ldots,T_k)$ \\
     \> \> \textbf{case} $\lambda x_1,\dotsc,x_n.\,\varphi'$: \> $\Env(x_1) := T_1; \dotsc; \Env(x_n) := T_n$ \\
     \> \> \> \textbf{return} $\Mc{\varphi'}{T_{n+1},\dotsc,T_k}$ \\
     \> \> \textbf{case} $\varphi' \, (\varphi_1,\dotsc,\varphi_n)$: \> \textbf{for} $i=1,\ldots,n$ \textbf{do} \\
     \> \> \> \quad \= \textbf{let} $\tau_1 \to \dotsb \to \tau_{k'} \to \grtype = \mathrm{type}(\varphi_i)$ \\
     \> \> \> \> $f_i := \{(T'_1,\dotsc,T'_{k'}) \mapsto \Mc{\varphi_i}{T'_1,\dotsc,T'_{k'}} \mid T'_i \in \sem{\tau_i}{\mathcal{M}}{}, i=1,\ldots,n \}$ \\
     \> \> \> \textbf{return} $\Mc{\varphi'}{f_1,\dotsc,f_n, T_1, \ldots, T_k}$ \\
     \> \> \textbf{case} $\sigma x.\, \varphi'$: \> $\Env(x) := \{(T_1,\dotsc,T_k) \mapsto \hat{\sigma_x}\}$ \\
     \> \> \> \textbf{repeat} \\
     \> \> \> \> $f := \Env(x)$ \\
     \> \> \> \> \textbf{for all} $(T'_1,\dotsc,T'_k) \in \mathrm{dom}(\Env(x))$ \textbf{do} \\
     \> \> \> \> \quad \= $\Env(x) := \Env(x)[(T'_1,\dotsc,T'_k) \mapsto \Mc{\varphi'}{T'_1,\dotsc,T'_k}]$ \\
     \> \> \> \textbf{until} $f = \Env(x)$ \\
     \> \> \> \textbf{return} $\Env(x)(T_1,\dotsc,T_k)$
   \end{tabbing}

\caption{Neededness-based evaluation for abstract higher-order fixpoint algebra.}
	\label{alg:ahofa-mc}
\end{algorithm}

We explain the algorithm's functionality by considering the different cases for its argument $\varphi$.
 Upon reaching a basic function symbol, \Mcname simply applies the semantics of this basic function to the arguments in the argument list.
When \Mcname reaches a variable $x$ and the value of that variable at argument $(T_1,\dotsc,T_k)$ is defined, then its value is returned. Otherwise, the variable must be fixpoint-bound, and 
\Mcname has discovered a new tuple of arguments at which the value of this fixpoint is needed. This value is initialised as $\hat{\sigma}_x$, which also registers $(T_1,\dotsc,T_k)$ in $\Env$. 
In this case the initial value is returned.
 
At a $\lambda$-abstraction, a number of arguments corresponding to the abstracted variables are transferred from the argument list to $\Env$, i.e.~they are now treated as bound variables. 
In an application $\varphi(\varphi_1,\ldots,\varphi_k)$, \Mcname computes, for each argument, its full semantics by a number of recursive calls to $\Mcname$\footnote{This can be done lazily, in case the argument is not needed or has been already computed. We omit the details for this in order to keep the presentation simple.}. The obtained values (as functions)
are then added to the list of arguments.
 
Upon reaching a fixpoint binder for variable $x$, \Mcname (re-)sets $\Env(x)$ to the singleton definition that initialises the value of the fixpoint at $(T_1,\dotsc,T_k)$ to the default value 
of $\hat{\sigma}_x$. Then, for each argument tuple that is already discovered as necessary for the value at $(T_1,\dotsc,T_k)$, the algorithm computes a new value. Note that, during this process
\Mcname can reach the variable case and discover new argument tuples. This procedure of updating the value at all known argument tuples is repeated until both no new arguments are discovered for 
one round, and the value of the fixpoint at each of the tuples agrees with that of the last round. If this has happened, the value of the last iteration at $(T_1,\dotsc,T_k)$ is returned.

% !TEX root =  main.tex

\paragraph*{Correctness.} The formal correctness proof for Alg.~\ref{alg:ahofa-mc} uses the following lemma which formalises the converse of Lemma~\ref{lem:monotonicity}. Take a term $\varphi$ that is typed with hypotheses 
$\Gamma, x^v$. Not only is it monotone (if $v = \varmon$), respectively antitone (if $v = \varant$) in the value of $x$. If the value of $\varphi$ also differs genuinely under two 
variable interpretations that only differ in $x$, then $x$ must occur freely in $\varphi$ and there are arguments to the value of $x$ on which this difference manifests itself. We write $x \sqsubset y$ to denote that $x \sqsubseteq y$ but $x \not = y$.

\begin{lemma}
\label{lem:semantic-reduction}
	Let $\mathcal{M}$ be a finite, and hence, complete lattice, 
	$\eta$ be a variable interpretation, 
	$f_1, f_2 \in \sem{\tau'}{}{}$ with $\tau' = \tau'_1\to\dotsb\to\tau'_k\to \grtype$ for some $\tau_1',\ldots,\tau_k'$, let $T_1,\dotsc,T_n$ be values with $T_i \in \sem{\tau_i}{}{}$ for $i=1,\ldots,n$, 
	$v \in \{\varmon,\varant\}$, and $\varphi$ be a $\muHO$ term such that $\Gamma, x^v\colon\tau'\vdash\varphi\colon\tau^{v_1}_1\to\dotsb\to\tau^{v_n}_n\to\grtype$. 
	If 
	  \[\sem{\varphi}{\mathcal{M}}{\eta[x \mapsto f_1]}(T_1,\dotsc,T_n) \sqsubset_\grtype \sem{\varphi}{\mathcal{M}}{\eta[x \mapsto f_2]}(T_1,\dotsc,T_n)\]
	then $x$ appears freely in $\varphi$, and there are $T'_1,\dotsc, T'_n$ such that 
	  \[f_1\, (T'_1,\dotsc,T'_n) \gtrless f_2\, (T'_1,\dotsc,T'_n)\]
	with ${\gtrless} = {\sqsubset_\grtype}$ if $v= \varmon$ and ${\gtrless} = {\sqsupset_\grtype}$ if $v= \varant$.
\end{lemma}

\begin{proof}
By induction on the structure of $\varphi$. Details are omitted. 
\end{proof}

Next we state correctness of Alg.~\ref{alg:ahofa-mc}. It is not hard to imagine local fixpoint iteration to be sound (for least fixpoints, resp.\ complete for greatest ones) since it clearly only performs
part of a global fixpoint iteration that is sound and complete according to Kleene's Theorem. For completeness one has to see that the value of a fixpoint function on some argument is determined solely by
the value of that function on its dependent arguments. Hence, it suffices to iterate on these until stability is reached. 

We write $\Mcc{\eta}{\varphi}{T_1,\dotsc,T_k}$ for the result of the call of Alg.~\ref{alg:ahofa-mc} on $\mathcal{M}$ with arguments $\varphi$ and $[T_1,\dotsc,T_k]$ such that $\Env$ satisfies 
$\eta = \eta_\Env$. 

\begin{theorem}
\label{thm:correctness}
	Let $\mathcal{M}$ be a finite, and, hence, complete lattice, $\eta$ be a variable interpretation, $\varphi$ be a term of type $\tau_1 \to \dotsb \to \tau_k \to  \grtype$, let
	$T_1,\dotsc,T_k$ be values with $T_i \in \sem{\tau_i}{\mathcal{M}}{}$, and $\Env$ be such that $\eta = \eta_{\Env}$. 
	Then $\Mcc{\eta}{\varphi}{T_1,\dotsc,T_k} = \sem{\varphi}{\mathcal{M}}{\eta}(T_1,\dotsc,T_k)$.
\end{theorem}
\begin{proof}
	The proof is, again, by induction on the structure of $\varphi$. Details are omitted. 
\end{proof}

A natural question that arises is the one after the time and space complexity of local higher-order fixpoint iteration. Two aspects
need to be considered here. First of all, it should be obvious that local evaluation cannot improve the worst-case. It is in fact not
hard to construct examples which a fixpoint term of higher-order such that its evaluation causes all argument values to be
explored. Consider the (order-1) term $\big(\nu F(x). F( (x \wedge \Diamond\neg x) \vee (\neg x \wedge \Box x))\big)(\myfalse)$ with
$\wedge,\vee,\neg,\Diamond,\Box$ interpreted in the usual way known from modal logic, over the powerset lattice induced by the
Kripke structure
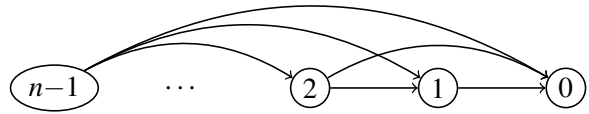
\begin{center}
\begin{tikzpicture}[every state/.style={inner sep=2pt,minimum size=3mm}, node distance=17mm, semithick]
  \node[state,shape=ellipse] (n1)              {$n{-}1$};
  \node                      (d) [right of=n1] {\ldots};
  \node[state]               (2) [right of=d]  {$2$};
  \node[state]               (1) [right of=2]  {$1$};
  \node[state]               (0) [right of=1]  {$0$};
   
  \path[<-] (0) edge (1) edge [bend right] (2) edge [bend right] (n1)
            (1) edge (2) edge [bend right] (n1)
            (2) edge [bend right] (n1);
\end{tikzpicture}
\end{center} 
Even though the term evaluates to $\{0,\ldots,n-1\}$, local fixpoint iteration will successively discover all $2^n$ arguments to the
first-order function $F$ before termination. It is also possible to extend this example to an arbitrary higher order.

Second, the question after the space and time complexity of Alg.~\ref{alg:ahofa-mc} cannot be answered without making assumptions on the representation of the lattice and 
the complexity of evaluating base functions. So far, no assumptions have been made explicitly, even though it is clear that such 
functions should at least be computable for otherwise Alg.~\ref{alg:ahofa-mc} would not be well-defined. A reasonable assumption is
that each base function of order $k$ can be evaluated in time and space that is at most $(k-1)$-fold exponential in the size of the
underlying lattice, with $0$-fold meaning polynomial and $(-1)$-fold meaning logarithmic. Logarithmic bounds may seem highly 
restrictive at first glance, but they make sense in cases where the underlying lattice is obtained as the powerset lattice of some
other structure, see the example above. If this assumption is met, then it is not too hard to see that Alg.~\ref{alg:ahofa-mc} runs
in time and space that is at most $k$-fold exponential with $k$ being the order of the input term. This also assumes that the lattice 
is given in a logarithmically sized representation. Otherwise, the complexity drops by one exponential.

% !TEX root =  main.tex

\section{Applications}
\label{sec:applications}

We present four applications of fixpoint evaluation in higher-order lattices using local fixpoint evaluation, and estimate how many computation steps can be saved compared to a na\"{\i}ve 
bottom-up and global fixpoint iteration.

\paragraph*{Constrained reachability problems.}
Reachability problems -- to decide whether some node is reachable from another in a directed graph -- are ubiquitous in computer science. In some applications, simple reachability is too coarse;
instead one wants to put constraints on the form of path under which the target node can be reached from the source, for instance in terms of distance, weight, shape or allowed sequence of 
edges. The latter can easily be formalised as a reachability problem constrained by some formal language. This has been investigated thoroughly for regular \cite{journals/tods/BarceloLLW12} 
and context-free languages \cite{Barrett:2000:FLC,la-reachpdl:2011,conf/fct/LangeL15} for applications in database theory, model checking, or in
static program analysis for heap-manipulating programs \cite{lev2000tvla,distefano2006local}. Little has been done for larger classes of languages.

We consider the context-sensitive language $L_{\mathsf{abc}} := \{a^n b^n c^n \mid n \ge 1 \}$ over the alphabet 
$\Sigma := \{ a,b,c \}$, and the problem whether for some given nodes $s,t$ of a directed, edge-labelled graph $G = (V,E)$ there is a path from $s$ to $t$ whose edge labels form a word in $L_{\mathsf{abc}}$.
Reachability problems can be interpreted in powerset lattices $(2^V, \subseteq)$ as they can be seen as least fixpoints of functions from sets of nodes to sets of
nodes. However, if $G$ is backwards-deterministic, i.e.\ for all $v,u,w \in V$, $a \in \Sigma$ we have $(v,a,u) \in E \wedge (w,a,u) \in E \Rightarrow v = w$, it is possible to formalise such problems over
a smaller lattice. 

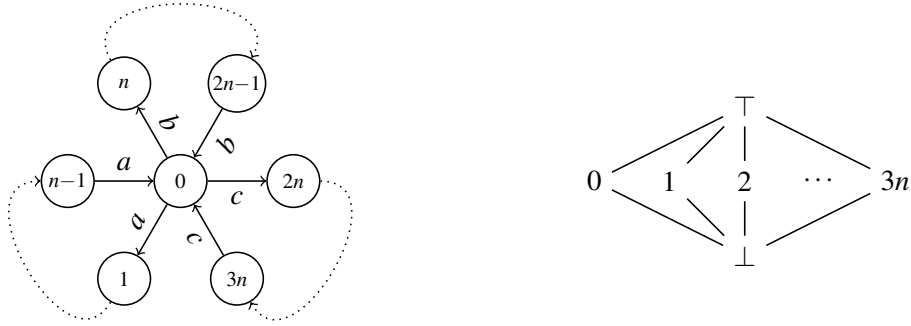
\begin{figure}
\begin{tikzpicture}[semithick,every state/.style={inner sep=1pt,minimum size=7mm,draw,font=\scriptsize}]

  \node[state] (m) at (0,0) {$0$};
  
  \foreach \x/\n in {1/{3n},2/{2n},3/{2n{-}1},4/n,5/{n{-}1},6/1} {
    \node[state] (s\x) at (240+\x*360/6:1.5cm) {$\n$};
  }
  
  \path[->] (m) edge                 node [sloped,above] {$a$} (s6)
                edge                 node [sloped,above] {$b$} (s4)
                edge                 node [sloped,below] {$c$} (s2)
;
  \path[<-] (m) 
                edge                 node [sloped,above] {$a$} (s5)
                edge                 node [sloped,below] {$b$} (s3)
                edge                 node [sloped,below] {$c$} (s1);
  \foreach \x/\y in {6/5,4/3,2/1} {
    \path[->,dotted,draw] (s\x) .. controls (240+\x*360/6:3cm) and (240+\y*360/6:3cm) .. (s\y);
  }

  \node (0) [right of=s2, node distance=4cm]             {$0$};
  \node (1) [right of=0] {$1$};
  \node (2) [right of=1] {$2$};
  \node (d) [right of=2] {$\cdots$};
  \node (e) [right of=d] {$3n$};
  \node (t) [above of=2] {$\top$};
  \node (b) [below of=2] {$\bot$};
  
  \foreach \x in {0,1,2,e}
    \path[-] (\x) edge (t) edge (b);

\end{tikzpicture}

\vspace*{-3mm}
\caption{Directed graph $G_n$ with $n > 1$ (left) for the language-constrained reachability problem example and corresponding base lattice (right).}
\label{fig:clover}
\end{figure}

Consider the graph $G_n$ depicted in Fig.~\ref{fig:clover} on the left. It contains a central state $0$ and around this three loops: an $a$-loop with $n-1$ states, a $b$-loop with $n$ states
and a $c$-loop with $n+1$ states. It is backwards-deterministic. 
Let $\grtype$ be interpreted by the lattice shown in Fig.~\ref{fig:clover} on the right. 
Intuitively, $\bot$ can be read as ``a path from source to the target has not been found yet'', and $\top$ signals that such a path has been found.

We use base functions $\mathcal{F} := \{ 0\colon \grtype, \ a,b,c\colon \grtype^\varmon \to \grtype, \ \mathsf{ite}\colon \grtype^\varboth \times \grtype^\varmon \to \grtype \}$ as follows.
The constant $0$ denotes the state $0$.
For any $v \in V$, $a(v)$ is the $a$-predecessor of $v$; likewise for $b$ and $c$. The value is $\bot$ if no such predecessor exists, in particular when applied to $\bot$. The value on
$\top$ is $\top$ itself. For instance, $c(0) = 3n$, $c(2n) = 0$, $c(v) = v-1$ if $2n < v \le 3n$, $c(\top) = \top$ and $c(v) = \bot$ otherwise.

In the powerset lattice $(2^V,\subseteq)$, $\mathsf{ite}$ could simply be interpreted as set union. However, here we interpret it as an \emph{if-then-else} in the following way.
Note that it is only monotonic in its second argument.
\begin{displaymath}
\mathsf{ite}(x,y) := \begin{cases}
\top &, \text{ if } x=0 \\
y &, \text{ otherwise }
\end{cases}
\end{displaymath}
Now let $\Var = \{ f,g\colon \grtype^\varmon \to \grtype, \ x\colon \grtype, \ F\colon (\grtype^\varmon \to \grtype)^\varboth \times (\grtype^\varmon \to \grtype)^\varboth \times \grtype^\varboth \to \grtype \}$ and consider the term 
\[
\varphi_{\mathsf{reach}} \enspace := \enspace \Big(\mu F(f,g,x).\, \mathsf{ite}\big(f(g(x)) , F(a \circ f, b \circ g, c(x))\big)\Big)\, \big(a, b, c(0)\big)
\]
where $\psi \circ \chi := \lambda x.\psi(\chi(x))$.

Using fixpoint unfolding and $\beta$-reduction one can see that the value of $\varphi_{\mathsf{reach}}$ becomes $\top$ when $a(b(c(0)) = 0$ or $a(a(b(b(c(c(0)))))) = 0$ or $a^3(b^3(c^3(0))) = 0$ 
and so on. Hence, evaluating $\varphi_{\mathsf{reach}}$ solves the reachability question ``is there a path from $0$ to $0$ under some word in $L_{\mathsf{abc}}$?''

We analyse how much computation power is being saved when computing the answer to the question of whether there is an $L_{\mathsf{abc}}$-path from
$0$ to $0$ in $G_n$. We compare four situations arising from the use of the standard powerset lattice vs.\ the optimised flat lattice of Fig.~\ref{fig:clover}, as well as local vs.\ 
global enumeration of all arguments to higher-order functions. 

Note that the three cycles in $G_n$ have lengths $n$, $n+1$ and $n+2$ respectively, which are always co-prime for each $n \ge 2$. 
Hence, the shortest $L_{\mathsf{abc}}$-path from $0$ to $0$ is the one that performs $(n+1)(n+2)$ many rounds on the $a$-cycle, then $n(n+2)$ rounds on the $b$-cycle and then $n(n+1)$ rounds on the
$c$-cycle. 

The following table shows the computational effort needed to evaluate $\varphi_{\mathsf{reach}}$ in terms of the number of arguments, resp.\ width of the table representing the function $F$. It also shows the space that is needed in order to represent one argument, i.e.\ a triple $(f,g,x)$ where $f,g$ are first-order functions 
and $x$ is a lattice element. Finally, in all cases the height of the table, i.e.\ the number of fixpoint iterations needed until $F$ stabilises, is in $\mathcal{O}(n^3)$.
\begin{center}
\begin{tabular}{l|c|c|c|c}
 & \multicolumn{2}{c|}{powerset lattice} & \multicolumn{2}{c}{flat lattice} \\ \hline
evaluation & \hspace*{3mm} global \hspace*{3mm} & \hspace*{5mm} local \hspace*{5mm} & \hspace*{3mm} global \hspace*{3mm} & \hspace*{5mm} local \hspace*{5mm} \\ \hline\hline
\rule[-3mm]{0pt}{8mm}width of table for $F$ 
  & $2^{2^{\mathcal{O}(n)}}$ 
  & $\mathcal{O}(n^3)$ 
  & $2^{\mathcal{O}(n\cdot\log n)}$ 
  & $\mathcal{O}(n^3)$ 
  \\ \hline
\rule[-3mm]{0pt}{8mm}size of arguments 
  & \multicolumn{2}{c|}{$2^{\mathcal{O}(n)}$} 
  & \multicolumn{2}{c}{$\mathcal{O}(n\log n)$} 
\end{tabular}
\end{center}

\begin{figure}
\hfill
\begin{minipage}{.34\textwidth}
\begin{align*}
P    \to\ &\blck \cdot B(S)  \\
S    \to\ &\spc \mid \spc \cdot S \\
B(I) \to\  &\epsilon      \\
     \mid\ &I \cdot \cod \cdot B(I) \\
     \mid\ &I \cdot \blck \cdot B(S \cdot I)
\end{align*}
\end{minipage}
\hfill
\begin{minipage}{.36\textwidth}
\begin{tabbing}
\blck \\
\spc \= \spc \= \cod \\
\spc \> \spc \> \blck \\
\spc \> \spc \> \spc \= \cod \\
\spc \> \spc \> \spc \> \cod \\
\spc \> \spc \> \cod \\[-8mm]
\end{tabbing}
\end{minipage}
\hfill
\begin{minipage}{.3\textwidth}
\begin{tikzpicture}[semithick, node distance=12mm]
  \node (0)                                 {$0$};
  \node (1) [right of=0, node distance=4mm] {$1$};
  \node (t) [above right of=0]              {$\top$};
  \node (b) [below right of=0]              {$\bot$};
  \node (n) [above right of=b]              {$n$};
  \node (d) [right of=1, node distance=6.5mm] {$\cdots$};  
  
  \path[-] (0) edge (t) edge (b)
           (1) edge (t) edge (b)
           (n) edge (t) edge (b);
\end{tikzpicture}
\end{minipage}
\caption{Higher-order grammar $\mathcal{G}_{\mathsf{ind}}$ (left); example word (middle); lattice $\mathcal{M}_w$ (right).}
\label{fig:hogrammar}
\end{figure}
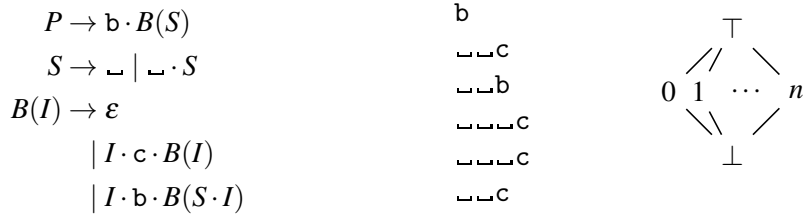

\paragraph*{Parsing of programming languages with indentation.}
Some programming languages like \textsc{Haskell} or \textsc{Python} use indentation as a syntax element. Such an effect can be described conveniently by the higher-order
grammar $\mathcal{G}_{\mathsf{ind}}$ shown in Fig.~\ref{fig:hogrammar} (left), over the terminal alphabet 
$\Sigma = \{ \blck, \cod, \spc \}$ (for ``block'', ``code'' and ``space''). 
We refrain from 
defining higher-order grammars \cite{Maslov74,Wand74} formally, since the technicalities needed for this small example are quite intuitive: 
$B(I)$ generates a block of code at indentation level $I$. The block can either be empty, contain one line followed by the rest of this
block at the same indentation level, or start a new block at a greater indentation level. The symbol $S$ is used to generate a sequence 
of space characters `\spc'. Finally, $P$ generates a program as a block at some initial indentation level. An example word, generated
from this grammar and formatted in order to visualise indentation best, is shown in the middle of Fig.~\ref{fig:hogrammar}.

Suppose a word $w = a_0 \ldots a_{n-1} \in \Sigma^*$ is given. This gives rise to an interpretation of the symbols in the grammar above 
as follows. Let $\mathcal{M}_w$ be the lattice shown in Fig.~\ref{fig:hogrammar} on the right. We use the following base functions, 
derived from the terminal symbols and the constructors in the higher-order grammar.
\begin{displaymath}
\Func_{\mathsf{ind}} = \{ \blck,\cod,\spc,\epsilon\colon \underbrace{\grtype^\varboth \times \grtype^\varboth \to \grtype}_{\tau}, \enspace \cdot,{\mid}\colon \tau^\varmon \times \tau^\varmon \to \tau, \enspace
 \mathsf{start},\mathsf{end}\colon \grtype\}
\end{displaymath}
Intuitively, \blck, \cod, \spc\ and $\epsilon$ are used to mark sections of the word in which the corresponding symbols (resp.\ the empty
word) occur. The top and bottom element of the lattice are used to signal true/false. The interpretation of these functions is therefore 
simply
\begin{displaymath}
\epsilon(i,j) = \begin{cases} \top &, \text{ if } i=j \\ \bot &, \text{ otherwise} \end{cases}
\quad \text{and} \quad
a(i,j) = \begin{cases} \top &, \text{ if } i+1=j \text{ and } a_i = a \\ \bot &, \text{ otherwise} \end{cases}
\end{displaymath}
for $a \in \Sigma$. The two constructors $\mid$ and $\cdot$ denoting disjunctive choice and concatenation in the higher-order grammar
are interpreted as follows.
\begin{displaymath}
(f {\mid} g)(i,j) = \begin{cases} \top &, \text{ if } f(i,j) = \top \text{ or } g(i,j) = \top \\ \bot &, \text{ otherwise} \end{cases}
\end{displaymath}
and
\begin{displaymath}
(f \cdot g)(i,j) = \begin{cases} \top &, \text{ if there is } h \text{ s.t. } f(i,h) = \top = g(h,j) \\ \bot &, \text{ otherwise} \end{cases}
\end{displaymath}
Finally, we need two constants $\mathsf{start}$ and $\mathsf{end}$ which are interpreted as $0$ and $n$, respectively.

The nonterminals in the higher-order grammar can be seen as (fixpoint) variables, hence we have 
$\Var = \{ P, S, I\colon \tau, B\colon \tau^\varmon \to \tau \}$. Then $\mathcal{G}_{\mathsf{ind}}$ immediately becomes a 
second-order term of $\muHO$ over $\Func_\mathsf{ind}$ and $\Var$, since recursion in grammars is captured by least fixpoints. 
The problem of evaluating $P(\mathsf{start},\mathsf{end})$ over $\mathcal{M}_w$ is then equivalent to parsing $w$ w.r.t.\ 
$\mathcal{G}_{\mathsf{ind}}$. 

Clearly, the space and time needed to evaluate $P(\mathsf{start},\mathsf{end})$ is dominated by the fixpoint iteration for $B$ as the only second-order variable. The number of possible arguments
to it is $2^{\mathcal{O}(n^2)}$. Local fixpoint iteration only discovers a fraction of these, though. Note that $B$ is initially evaluated on $S$, and -- when the recursive is called on argument 
$I$ -- it needs the values of $B$ on $I$ itself as well as on $S \cdot I$. Hence, it only ever discovers $S, S^2, S^3, \ldots$. Moreover, it is not hard to see that $S^k$ maps two positions $(i,j)$
of an underlying word $w = a_0\ldots a_{n-1}$ to $\top$, if $j-1-i \ge k$ and $a_h = \spc$ for $h=i,\ldots,j-1$. Hence, the number of possible arguments to $B$ discovered in this way is bounded
by $n-1$ and so, again, local higher-order fixpoint iteration realises an exponential reduction in space complexity in this example.

\paragraph*{Model checking Higher-Order Fixpoint Logic.}
Fixpoints play a fundamental role in model checking, where properties of the runtime behaviour of programs are typically expressed in temporal logics, the most prominent of which are LTL and CTL. 
Fixpoints are used there to express limit behaviour as in reachability, safety and fairness \cite{ICALP::EmersonC1980}. The true power of fixpoints is unleashed in logics that extend modal logic by
extremal fixpoint quantifiers like the well-known modal $\mu$-calculus \cite{Kozen83}. A lesser known extension of this is Higher-Order Fixpoint Logic (\hfl) \cite{DBLP:conf/concur/ViswanathanV04}, a
highly expressive specification logic that mixes modal logic, a typed $\lambda$-calculus and fixpoint quantifiers. Its model checking problem is decidable over finite transition systems, albeit
$k$-fold exponential in the order of involved function types \cite{als-mchfl07,BLL:RP17}.

We refer to the literature for a self-contained definition of \hfl \cite{DBLP:conf/concur/ViswanathanV04}. With the preliminary work on abstract higher-order fixpoint algebra in Sect.~\ref{sec:prelim}
we can simply present \hfl as a special instantiation of this algebra. The base type $\grtype$ is interpreted as the powerset lattice $(2^\States, \subseteq)$ of the state set of a transition system 
with edge labels from some set $\Actions$ and propositional labels from some set $\Prop$. This gives rise to an interpretation of all higher-order types as functions on sets of states; a function of type $\grtype^\varmon \to \grtype$ for instance
is known as a (monotonic) \emph{predicate transformer}.

The set of ground functions then is $\mathcal{F}:= \{ \wedge, \vee\colon \grtype^\varmon \times \grtype^\varmon \to \grtype, \enspace \neg\colon \grtype^\varant \to \grtype \} \cup \{ \mydia{a}, \mybox{a}\colon \grtype^\varmon \to \grtype \mid a \in \Actions \} \cup \{p \colon \grtype \mid p \in \Prop\}$
reflecting the Boolean, modal and propositional parts of the logic. %We use infix notation for the binary operators $\wedge$ and $\vee$ as usual.
Well-typed terms of the higher-order fixpoint algebra over this $\mathcal{F}$ are exactly the formulas of \hfl; and the standard semantics is the same as the one derived from the generic semantics 
in Sect.~\ref{sec:prelim} for this set of terms.

Consider the following formula describing the property $\varphi =$ ``there is an infinite $b$-path such that the $i$-th node on 
this path is the start of an $a$-path of length $2^{i}$ ending in a $p$-node, for any $i \ge 0$'', as well as the family of 
transition systems $\mathcal{T}_n$ on the right. 

\begin{minipage}{.42\textwidth}
\begin{displaymath}
\varphi := \big(\nu F.\,\lambda f.\,(f\, p) \wedge \mydia{b}(F (f \circ f))\big)\, (\lambda x.\mydia{a}x)
\end{displaymath}
\end{minipage} \hfill 
\begin{minipage}{.45\textwidth}
  \begin{tikzpicture}[node distance=13mm, every state/.style={inner sep=2pt,minimum size=3mm}]
    \node[state,label=90:{$p$}] (0)              {$0$};
    \node[state]                (1) [right of=0] {$1$};
    \node[state]                (2) [right of=1] {$2$};
    \node                       (d) [right of=2] {$\cdots$};
    \node[shape=ellipse,draw,inner sep=1pt] (e) [right of=d] {$n{-}1$};
    
    \path[->] (0) edge [bend left] node [above] {$a$} (1)
                  edge [loop left] node [left]  {$b$} ()
              (1) edge [bend left] node [above=-2pt] {$a,b$} (0)
                  edge [bend left] node [above] {$a$} (2)
              (2) edge [bend left] node [above] {$a$} (1)
                  edge [bend left] node [below,very near start] {$b$} (0)
                  edge [bend left] node [above] {$a$} (d)
              (d) edge [bend left] node [above] {$a$} (2)
                  edge [bend left] node [above] {$a$} (e)
              (e) edge [bend left] node [above] {$a$} (d)
                  edge [bend left] node [below,very near start] {$b$} (0);
  \end{tikzpicture}
\end{minipage}

\noindent
where $\varphi \circ \psi := \lambda x.\,\varphi\, (\psi\, x)$, over $\Var = \{ x\colon \grtype, \ f\colon \grtype^\varmon \to \grtype, \ F\colon (\grtype^\varmon \to \grtype)^\varmon \to \grtype \}$.
We use $F$ in the following to abbreviate the subformula $\nu F.\lambda f.(f\, p) \wedge \mydia{b}(F\, (f \circ f))$.

Only state $1$ satisfies $\varphi$. Now note that $F$ is a second-order fixpoint
taking as arguments a term interpreted as a first-order function of the kind $2^{[n]} \to 2^{[n]}$. Hence, 
even for $n=2$, there already are 256 of them, 
and na\"{\i}ve fixpoint iteration would tabulate all of them first before computing the values of $F$ on them. On the other hand, all that is needed is $F$'s value on functions $\myddia{a}^{2^i}$ where 
$\myddia{a}(S) = \{ t \in [n] \mid \exists t \in [n]$ s.t.\ $\Transition{s}{a}{t} \}$. The following puts the number of such different
functions which are being discovered by local fixpoint iteration in relation to the number of otherwise possible function argument.
\begin{center}
\begin{tabular}{r||c|c|c|c|c|c|c}
$n$                                   & 2   & 3                & 4 & 5 & 6 & 7 & \ldots \\ \hline
possible arguments to $F$             & 256 & $1.6 \cdot 10^6$ & $1.8\cdot 10^{19}$ $1.5\cdot10^{48}$ & $3.9\cdot 10^{115}$ & \ldots & \ldots & \ldots \\ \hline
discovered in local iteration & 2   & 2                & 2 & 3 & 3 & 4 & \ldots
\end{tabular}
\end{center}
The numbers can be verified either through manual computation of the functions $\myddia{a}^{2^i}$ for $i=0,1,\ldots$ on each 
$\mathcal{T}_n$ or using the implementation of Alg.~\ref{alg:ahofa-mc} mentioned in the conclusion below.

\newcommand{\zero}{{\bf{0}}}
\newcommand{\one}{{\bf{1}}}
\newcommand{\two}{{\bf{2}}}

  \paragraph*{Abstract interpretation of functional languages.}
Strictness analysis for (lazy) functional languages tries to figure out
whether an argument to a function must always be evaluated. In this case
compilers may force the evaluation of the argument thus saving space and
time to create closures and allowing for parallelisation.
Strictness analysis may be formulated as an abstract
interpretation as e.g.\ in \cite{burn1986strictness}. A function
$f\colon D\times D\times\dotsm\times D\rightarrow D$ is \emph{strict}
in its $i$-th argument, when $f(d_1,d_2,\ldots,d_{i-1},\bot_D,d_{i+1},\ldots d_k)=\bot_D$ for
a concrete base domain $D$. As this may be uncomputable, in \cite{burn1986strictness},
functions are interpreted \emph{abstractly} over the
domain $\two \enspace := \enspace \{ \zero, \one \}$ (with $\zero\sqsubseteq \one$), where $\zero$
means \emph{definitely undefined}, and where $\one$ means \emph{might be defined}.
Examples of abstract interpretations of common base values are (for $x,y,z\in\two$),
\begin{itemize}
\item constants of base domains such as integers or boolean values are abstracted to
  $\one$ (not \emph{undefined});
\item first-order functions such as addition are strict in all arguments, e.g.,
  $x+y=\zero$ unless $x=y=\one$;
\item if-then-else: $\mathsf{ite}(x,y,z) = x\wedge (y\vee z)$, where elements of
  $\two$ are read as boolean values. If-then-else might only be defined if both
  the condition and at least one of the then-else arguments might be defined. Otherwise it is
  definitely undefined. Note that, in this example, we use $\mathsf{ite}$
  in the traditional sense of functional programming.
\end{itemize}
As an example of the application of \Mcname to the abstract interpretation of
functional languages, we choose 
$\mathcal{F} := \{ \mathsf{ite} \colon \grtype^\varmon \times \grtype^\varmon \times \grtype^\varmon  \to \grtype \}$
and 
$\Var = \{ x\colon \grtype,\ f\colon \grtype^\varmon \to \grtype, p\colon \grtype^\varmon \to \grtype,\ I\colon (\grtype^\varmon \to \grtype)^\varmon \times
(\grtype^\varmon \to \grtype)^\varmon \times \grtype^\varmon \to \grtype \}$. Consider
the term $\varphi \enspace := \enspace \mu I(f,p,x).\, \mathsf{ite}(p(x),I(f,p,f(x)),x)$.
It essentially describes an iterated application of some function $f$ until a predicate $p$ holds.
In order to show that $\varphi$ is strict in $x$ for given functions $f_0$ and $p_0$, one needs to evaluate
$\varphi(f_0,p_0, \zero)$ by fixpoint unfolding and $\beta$-reduction. If $p_0(\zero) = \zero$, that is, the
termination predicate is itself strict, then \Mcname terminates in one step proving strictness of
$\varphi$ in its third argument. If $p_0(\zero) = \one$, that is, $p$ is essentially a constant
true or constant false predicate, we need to evaluate
$\one \wedge (\varphi(f_0,p_0,f_0(\zero)) \vee \zero) = \varphi(f_0,p_0,f_0(\zero))$ next.
If $f_0(\zero)=\zero$, that is, $f_0$ is strict itself, we have reached a fixpoint and can conclude strictness in
$x$ as well. If, however, $f_0(\zero)=\one$, we obtain an overall result of $\one$, not showing strictness in $x$.
This is plausible for constant functions $f$ and $p$.
Using local iteration this is in fact the only computation that takes place, whereas a na\"{\i}ve global
fixpoint computation would start by tabulating all possible triples of type
$(\grtype^\varmon \to \grtype)^\varmon \times (\grtype^\varmon \to \grtype)^\varmon \times \grtype^\varmon$,
which, for the lattice $\two$, amounts to $4\cdot 4\cdot 2=32$ table columns.

% !TEX root =  main.tex

\section{Limitations of Neededness Analysis and Optimisation}
\label{sec:operands}
As mentioned in the introduction to Sect.~\ref{sec:algo}, Algorithm \Mcname does not use local evaluation on operand-side subterms but rather computes their value fully. If such an operand has a function type, its value on all its arguments might not be needed either. Consider the first example from Sect.~\ref{sec:applications} about formal-language
constrained reachability problems. Clearly, the values of the order-$1$-functions stored in the parameters $f$ and $g$ are not needed at most arguments. Hence, computing 
their value fully appears to be wasteful. 

Algorithm \Mcname computes values of operand-side subterms fully due to the termination criterion for the computation of fixpoint terms: iteration stops when both no new argument tuples 
have been discovered during a round of the repeat-loop computing the semantics, and the value of the fixpoint in question is stable on all existing tuples. This, of course, requires 
some way of deciding whether a discovered argument is actually new. Going back to the example in Sect.~\ref{sec:applications}, Algorithm \Mcname successively discovers the argument tuples 
$[a, b, c(0)], [a^2, b^2, c^2(0)], \dotsc$ Eventually, these argument tuples begin to repeat, which is when the loop terminates. However, deciding whether e.g.~$[a^i, b^i, c^i(0)]$ 
is the same argument tuple as $[a^j, b^j, c^j(0)]$ requires knowing the value of the function type arguments at all \emph{their} arguments. One could assume that it is enough to know just their 
value on arguments actually needed in the iteration, but this approach fails readily: already for $[a, b , c(0)]$ and $[a^2, b^2, c^2(0)]$, for $n \geq 2$, we see that $c^i(0) = 3n-(i+1)$ for 
$i \in \{1,2\}$, whence $a^i(b^i(c^i(0))) = \bot$ for either $i$, and, in fact, all $i \leq n(n-1)(n-2)$, since these differ on hitherto undiscovered arguments. Hence, any algorithm that tests equality of function type arguments only on tuples already 
identified as necessary for the computation must fail here. Moreover, since the base functions $a,b,c$ are actually interpreted, instead of e.g.\,tree constructors as in the case of higher-order model checking, a simple flow analysis (e.g. $0$-CFA) fails to detect which functions are duplicates unless one also inspects the behavior of the base functions. Hence, safe approach to avoid the  error sketched above is to compute values of argument-side functions -- which are necessarily not of the
highest type order occurring in the term under consideration -- in full. 

However, this does not mean that this is always necessary. In the example from Sect.~\ref{sec:applications}, one can readily see that the value of e.g.~$f$ will always be $\bot$ on all arguments 
that are not in $\{0,\dotsc,n-1\}$, since $f$ only contains powers of $a$. This kind of domain-specific approach, together with e.g. flow analysis and the choice of an appropriate lattice, could be used 
to cut down the amount of computations necessary.

% !TEX root =  main.tex

\section{Conclusion}
\label{sec:concl}

We have lifted the notion of local fixpoint iteration, resp.\ neededness analysis, for the evaluation of first-order fixpoint
functions to fixpoint functions at arbitrary higher order. For generality purposes we have defined an abstract algebra \muHO
combining a simply typed $\lambda$-calculus over (possibly higher-order) base functions with fixpoints at arbitrary type orders. 
The examples in Sect.~\ref{sec:applications} show that this can vastly reduce the number of values that are being computed in fixpoint iterations, compared to the na\"{\i}ve global approach. 

A conceptual implementation of \muHO and Alg.~\ref{alg:ahofa-mc} is available.\footnote{\url{https://github.com/muldvarp/LocalHOFPIter}} 
It does not compete with specialised tools like higher-order model checkers but rather focuses on displaying the effect that
local fixpoint iteration has in comparison to global iteration for higher-order fixpoints.

Work on fixpoint iteration for higher-order functions can be continued in several directions. The most pressing issue is an
extension to \emph{fully} local fixpoint iteration, which would also employ local evaluation at orders beneath the top one, 
bearing in mind the obstacles to overcome which have been discussed in Sect.~\ref{sec:operands}. Significant progress on this front likely requires giving up the full genericity of the algorithm. For example, many intersection-type based HORS model checkers (e.g.\,\cite{DBLP:conf/csl/BroadbentK13,DBLP:conf/popl/RamsayNO14}) require backwards reasoning alongside the base functions. For example, acceptance of an automaton in a node of a tree depends on its children, i.e.\,the arguments to the tree constructor in question, and the relationship is readily available. Conversely, in the present form, our algorithm makes no assumptions on the (behavior of) the base functions, whence it can not infer which values of a given argument might yield a desired function value.

The acute reader may have wondered why \muHO does not feature operators $\sqcup, \sqcap$ for suprema and infima at arbitrary
types. It would in fact be possible to add these, and algorithm \Mcname can be extended accordingly to handle them just like
other base functions are being handled. They are not included in the syntax of \muHO here for the following reason: when $\sqcup, \sqcap$
are present in the syntax one would expect the distributivity laws like 
$\varphi \sqcup (\psi \sqcap \chi) \equiv (\varphi \sqcap \psi) \sqcup (\varphi \sqcap \chi)$ to hold. But in arbitrary lattices, 
such laws do not necessarily hold; they only do in distributive latttices. In order not to confuse the issue or make false assumptions
we therefore prefer to introduce $\sqcup,\sqcap$ as base functions when necessary and appropriate. This prevents us from restricting 
the semantics of \muHO terms to distributive lattices only. Note that the lattices depicted in Figs.~\ref{fig:clover} and 
\ref{fig:hogrammar} are not distributive.

Algorithm \Mcname makes no assumptions on the order in which needed arguments are evaluated. In data flow analyses, giving 
precedence to the arguments in the form of heuristics has turned out to be beneficial for efficiency purposes, c.f.\ 
\cite[Chp.~6]{nielson2015principles}. It remains to be seen whether such heuristics can be extended to higher orders as well.

Most static program analyses in abstract interpretation work with rather rich lattices as base domains which cannot be cast into
the scheme of a simply typed $\lambda$-calculus over a single base type $\grtype$ as it is used here. We remark, though, that
an extension to a many-sorted logic over several base types is straight-forward, not only regarding the type system but, most
importantly, algorithm \Mcname. The same holds for product types on the right of function arrows. It then remains to be seen
how far the type system can be enriched without seriously interfering with the ability to evaluate higher-order fixpoints 
locally.

\bibliographystyle{eptcs}
\bibliography{literature}

\newpage
%\appendix

%\input{proofs}

\end{document}